\documentclass{article}%
\usepackage[margin=1.65 in]{geometry}
\usepackage{amsmath}
\usepackage[linkcolor=blue,colorlinks=true]{hyperref}
\usepackage{amsfonts}
\usepackage{amssymb}
\usepackage{graphicx}
\usepackage{diagbox}
\usepackage{subfig}
\usepackage{epsfig}%
\usepackage{color}
\usepackage{multirow}
\usepackage{setspace}
\usepackage{hyperref}
\usepackage{enumitem}
\usepackage{hhline}
\usepackage{algorithm}
\usepackage{algorithmic}
\usepackage{pbox}
\usepackage{booktabs}
\usepackage{array}
\usepackage[flushleft]{threeparttable}
\usepackage{authblk}
\usepackage{float}
\usepackage{supertabular}
\usepackage[title]{appendix}
\usepackage{authblk}
\newtheorem{theorem}{Theorem}

\newtheorem{lemma}[theorem]{Lemma}

\newenvironment{proof}[1][Proof]{\noindent\textbf{#1.} }{\ \rule{0.5em}{0.5em}}
\begin{document}
\doublespacing

\title{\textbf{A Bayesian Semiparametric Gaussian Copula Approach to a Multivariate Normality Test}}   


\author[1]{Luai Al-Labadi\thanks{{\em Corresponding author:} luai.allabadi@utoronto.ca}}

\author[2]{Forough Fazeli Asl\thanks{forough.fazeli@math.iut.ac.ir}}

\author[2]{Zahra Saberi\thanks{ z\_saberi@cc.iut.ac.ir}}

\affil[1]{Department of Mathematical and Computational Sciences, University of Toronto Mississauga, Mississauga, Ontario L5L 1C6, Canada.}
\affil[2]{Department of Mathematical Sciences, Isfahan University of Technology, Isfahan 84156-83111, Iran.}

\date{}
\maketitle

\pagestyle {myheadings} \markboth {} {A BSPGC approach to a MVN test.}

\begin{abstract}
In this paper, a Bayesian semiparametric copula approach is used to model the underlying multivariate distribution $F_{true}$. First, the Dirichlet process is constructed on the unknown marginal distributions of $F_{true}$. Then a Gaussian copula model is utilized to capture the dependence structure of $F_{true}$.  As a result, a Bayesian multivariate normality test is developed by combining the relative belief ratio and the Energy distance. Several interesting theoretical results of the approach are derived. Finally, through several simulated examples and a real data set, the proposed approach reveals excellent performance.
\par

 \vspace{9pt} \noindent\textsc{Keywords:} Dirichlet process, Energy distance, Multivariate normality test, Relative belief inferences, Semiparametric Gaussian copula model.

 \vspace{9pt}


\end{abstract}
	
\section{Introduction}
Semiparametric copulas are useful tools in multivariate data analysis. They are used for modelling a multivariate distribution whose dependence structure is induced by a known copula and whose marginal distributions are estimated; see, for example, Sancetta and Satchell (2004), Segers et al. (2014) and the references therein.  We point out to the interesting  work of Rosen and Thompson (2015) who proposed a semiparametric methodology for modeling a multivariate distribution whose dependence structure is induced by a Gaussian copula and whose marginal distributions are estimated nonparametrically via mixtures of B-spline densities.  The authors take a Bayesian approach, using Markov chain Monte Carlo methods for inference.

In the present paper, a Bayesian Semiparametric copula approach based on the Dirichlet process and the Gaussian copula is proposed to model the underlying multivariate distribution $F_{true}$. In addition, recognizing that  many recent applications of research are developed based on the assumption of multivariate normality (Fernandez, 2010 and Zhu et al., 2014), a test  to assess this assumption is developed. Recent procedures tackling this problem can be found in Kim and Park (2018), Madukaife and Okafor (2018), Henze and Visagie (2019) and Al-Labadi et al. (2019a). We highlight that, while most available works in the area of the hypothesis testing using copula approaches are related to assess  independence  (Genest and R\'{e}millard, 2004; Kojadinovic and Holmes, 2009; Medovikov, 2016; Belalia et al., 2017), the proposed test is Bayesian and considers modeling the dependence structure and the marginal behaviors of the data separately to assess the multivariate normality assumption. Briefly, all univariate marginal distributions of $F_{true}$ are assumed to have the Dirichlet process to define posterior-based and prior-based models of $F_{true}$.  A Gaussian copula model is then utilized to induce the dependence structure. The test follows by comparing the concentration of the posterior-based model to the concentration of the prior-based model about the family of multivariate normal distributions (hypothesized model) via the so-called relative belief ratio.  In this comparison, a Bayesian counterpart of the Energy distance is developed. The proposed test is easy to implement with a powerful performance and allows to state evidence for or against the null hypothesis. Also, unlike the test presented in Al-Labadi et al. (2019a),  which is restricted to assess the family of multivariate normal distributions for the hypothesized model, the proposed test can be extended to assess all families of multivariate distributions (model checking problems).

The rest of the paper is structured as follows. A relevant background containing some definitions and generic properties are reviewed in Section 2. In Section 3, a Bayesian semiparametric Gaussian copula approach based on the Dirichlet process is proposed for modeling multivariate distributions. The choice of the hyperparameter of the Dirichlet process and the estimation method of the parameter of the Gaussian copula are discussed in Section. In Section 5, a Bayesian multivariate normality (MVN) test based on the proposed approach and the Energy distance is developed. The main steps of a computational algorithm to implement the MVN test are outlined in Section 6. The performance of the approach and its application to the MVN test is clarified through some simulation studies and a real data example in Section 7. The results show that the proposed test works well in all covered cases and it is very powerful. Finally, Section 8 concludes the paper with a summary of the results. Some notations related to the Section 7 are given in the Appendix.

\section{Relevant background}
\subsection{Copula-based Model}

In multivariate analysis, copula models are introduced by Sklar (1959) as a common tool to model multivariate distributions. Following Nelsen (2006),
an $m$-dimensional copula ($m$-copula) is a nondecreasing and right continuous function $C$ from $[0,1]^{m}$ into $[0,1]$ such that, for every
$\mathbf{u}=(u_{1},\ldots,u_{m})\in [0,1]^{m}$
\begin{enumerate}[nolistsep,label=(\roman*),leftmargin=\parindent,align=left,labelwidth=\parindent,labelsep=0pt]
\item
$C(u_{i},\ldots,u_{i-1},0,u_{i+1},\ldots,u_{m})=0$, for $i=1,\ldots,m$ ($C$ is grounded).
\item
$C(1,\ldots,1,u_{i},1,\ldots,1)=u_{i}$, for $i=1,\ldots,m$ ($C$ has margins).
\item
\hspace{0.1mm} For every $\mathbf{v}=(v_{1},\ldots,v_{m})\in [0,1]^{m}$ such that $u_{i}\leq v_{i}$ for all $i$, the $C$-volume
$V\left([\mathbf{u},\mathbf{v}]\right)\geq 0$ ($C$ is $m$-increasing), where $m$-cube $[0,1]^{m}$ is $m$ product of $[0,1]$ and $[\mathbf{u},\mathbf{v}]$ is the $m$-box $[u_{1},v_{1}]\times\cdots\times[u_{m},v_{m}]$.
\end{enumerate}

From (iii), it is obvious that every $m$-copula $C$ is nondecreasing in each variable and satisfies in the Lipschitz condition. That is, for every point $\mathbf{u}, \mathbf{v}\in [0,1]^{m}$,
\begin{equation}\label{Lipschitz}
\big| C(\mathbf{u})-C(\mathbf{v}) \big| \leq \sum_{i=1}^{m}|u_{i}-v_{i}|.
\end{equation}

\noindent Hence, any $m$-copula $C$ is uniformly continuous on $[0,1]^{m}$. 

The following key theorem of Sklar (1959) illustrates the role of the $m$-copulas to model the multivariate distribution functions through their univariate margins.

\begin{theorem}\label{sklar's Thm}
\textbf{(Sklar's theorem)} Let $F$ be an $m$-variate distribution function with marginal distribution functions $F_{1},\ldots,F_{m}$. Then there exists an $m$-copula $C$ such that for all $\mathbf{t}\in \overline{\mathbb{R}}^{m}$
\begin{equation}\label{sklar}
F(t_{1},\ldots, t_{m})=C\left(F_{1}(t_{1}),\ldots, F_{m}(t_{m})\right),
\end{equation}
where $\overline{\mathbb{R}}^{m}$ is $m$ product of the extended real line $[-\infty,\infty]$. If $F_{1},\ldots,F_{m}$ are all continuous, then $C$ is unique and can be written as
 \begin{equation*}
 C(\mathbf{u})=F\left(F^{-1}_{1}(u_{1}),\ldots, F_{m}^{-1}(u_{m})\right),
 \end{equation*}
for any $\mathbf{u}\in[0,1]^{m}$ where $F^{-1}_{i}(u_{i})=\inf\lbrace t\in \mathbb{R}|\, F_{i}(t)\geq u_{i}\rbrace$, otherwise; $C$ is uniquely determined on
$Ran (F_{1})\times \cdots \times Ran (F_{m})$, where $Ran (F_{i})$ denotes the range of the distribution function $F_{i}$ for $i=1,\ldots,m$. Conversely, if $C$ is an $m$-copula and $F_{1},\ldots,F_{m}$ are distribution functions, then $F$, defined by \eqref{sklar}, is an $m$-variate distribution function with marginal distribution functions $F_{1},\ldots,F_{m}$.
\end{theorem}

In general, the dependence structure of multivariate distributions is modeled by the copula. For this purpose, some families of copulas have been developed. See, for example, the work of Joe (1997), Trivedi and Zimmer (2005), and Nelsen (2006). One instance of special interest is the family of \emph{Gaussian copulas}. Beside that it satisfies both Fr\'{e}chet-Hoeffding lower and upper bounds (Nelsen 2006, Theorem 2.10.2), it has only one dependence parameter restricted to the symmetric interval $[-1,1]$, which makes it simple to apply. Formally, let $\Phi^{-1}$ be the inverse of the cumulative distribution function (cdf) of the univariate standard normal distribution and $\Phi_{R}$ be the cdf of the multivariate normal distribution with zero mean vector $\mathbf{0}_{m}$ and correlation matrix $R=(r_{ij})$ for $1\leq i,j\leq m$, then the family of Gaussian copulas is defined by
\begin{equation*}\label{gaussian-cop}
\left\lbrace C_{R}(\mathbf{u}):=\Phi_{R}( \Phi^{-1}(u_{1}),\ldots, \Phi^{-1}(u_{m}))|\hspace{.2cm} \, r_{ij}\in[-1,1]\right\rbrace,
\end{equation*}
\noindent where $\mathbf{u}\in [0,1]^{m}$. Following Chen et al. (2006), assume that, for any $\mathbf{u}\in [0,1]^{m}$, $C(\mathbf{u})=C_{R}(\mathbf{u})$, then the multivariate distribution $F$ is of a \emph{semiparametric Gaussian copula model} $F(t_{1},\ldots,t_{m})=C_{R}(F_{1}(t_{1}),\ldots, F_{m}(t_{m}))$, with unknown parameter $R$ and unknown marginal cdf $F_{i}$, for $i=1\ldots m$. A detailed discussion about the semiparametric Gaussian copula model will be presented in Section 3 based on using the Dirichlet process.

The following algorithm shows the steps of generating a sample of random vectors from an $m$-variate distribution $F$ with marginal cdf's $F_{1},\ldots,F_{m}$ using a Gaussian copula model with correlation matrix $R$ .
\begin{algorithm}
\caption{Generating a sample from an $m$-variate distribution $F$ using Gaussian copula model}\label{alg1}
\begin{enumerate}
\item Generate a random vector $(y_{1},\ldots,y_{m})$ from an $m$-variate normal distribution with mean vector $\mathbf{0}_{m}$ and correlation matrix of $R$. \textbf{\smallskip}
\item Compute $u_{i}=\Phi(y_{i})$, for $i=1,\ldots,m$. \textbf{\smallskip}
\item Compute vector $(x_{1},\ldots,x_{m})$ such that $x_{i}=F^{-1}_{i}(u_{i})$, for $i=1,\ldots,m$.
\item Repeat steps 1-3 for $k$ times to generate a sample of size $k$ from distribution $F$.\textbf{\smallskip}
\end{enumerate}
\end{algorithm}
\subsection{Dirichlet Process}
The Dirichlet process prior, introduced by Ferguson (1973), is the
most commonly used prior in Bayesian nonparametric inferences. A remarkable collection of nonparametric inferences have been devoted to this prior. Here we only present the most relevant
definitions and properties of this prior. Consider a space
$\mathfrak{X}$ with a $\sigma$-algebra $\mathcal{A}$ of subsets of
$\mathfrak{X}$, let $H$ be a fixed probability measure on $(\mathfrak{X}%
,\mathcal{A}),$ called the \emph{base measure}, and $a$ be a positive number,
called the \emph{concentration parameter}. A random probability measure
$P=\left\{  P(A):A\in\mathcal{A}\right\}  $ is called a Dirichlet process on
$(\mathfrak{X},\mathcal{A})$ with parameters $a$ and $H,$ denoted by $P\sim
{DP}(a,H),$ if for every measurable partition $A_{1},\ldots,A_{k}$ of
$\mathfrak{X} $ with $k\geq2\mathfrak{,}$ the joint distribution of the vector
$\left(  P(A_{1}),\ldots\,P(A_{k})\right)$ has the Dirichlet distribution with parameter
$aH(A_{1}),\ldots,$ $aH(A_{k})$. Also, it is assumed that
$H(A_{j})=0$ implies $P(A_{j})=0$ with probability one. Consequently, for any
$A\in\mathcal{A}$, $P(A)\sim$ beta$(aH(A),a(1-H(A)))$,
${E}(P(A))=H(A)\ $and ${Var}(P(A))=H(A)(1-H(A))/(1+a).$ Accordingly, the base measure $H$
plays the role of the center of $P$ while the concentration parameter $a$ controls the variation of
 $P$ around  $H$. One of the most well-known
properties of the Dirichlet process is the conjugacy property. That is, when the sample $x=(x_{1},\ldots,x_{n})$
is drawn from $P\sim DP(a,H)$, the posterior distribution of $P$ given $x$,
denoted by $P^{\ast}$, is also a Dirichlet process with
concentration parameter $a+n$ and base measure
\begin{equation}\label{pos base measure}
H^{\ast}=a(a+n)^{-1}H+n(a+n)^{-1}F_{n},
\end{equation}
where $F_{n}$ denotes the empirical cumulative distribution function (cdf) of the sample
$x$. Note that, $H^{\ast}$ is a convex combination
of the base measure $H$ and the empirical cdf $F_{n}$. Therefore, $H^{\ast}\rightarrow H$ as
$a\rightarrow\infty$ while $H^{\ast}\rightarrow F_{n}$ as $a\rightarrow0.$ A guideline
about choosing the hyperparameters $a$ and $H$ will be covered in Section 4. Following Ferguson (1973), $P\sim{DP}(a,H)\ $can be represented as
\begin{equation}
P=\sum_{i=1}^{\infty}L^{-1}(\Gamma_{i}){\delta_{Y_{i}}/}\sum_{i=1}^{\infty
}{{L^{-1}(\Gamma_{i})}}, \label{series-dp}%
\end{equation}
where $\Gamma_{i}=E_{1}+\cdots+E_{i}$ with $E_{i}\overset{i.i.d.}{\sim}%
$\ exponential$(1),Y_{i}\overset{i.i.d.}{\sim}H$ independent of the
$\Gamma_{i},L^{-1}(y)=\inf\{x>0:L(x)\geq y\}$ with $L(x)=a\int_{x}^{\infty
}t^{-1}e^{-t}dt,x>0,$ and ${\delta_{a}}$ the Dirac delta measure. The series representation
(\ref{series-dp}) implies that the Dirichlet process is a discrete probability
measure even for the cases with an absolutely continuous base measure $H$. Note
that, by imposing the weak topology, the support of the Dirichlet process could
be quite large. Recognizing the complexity when working with \eqref{series-dp}, Zarepour and Al-Labadi (2012) proposed the following finite representation as an efficient method to simulate the Dirichlet process. They showed that the Dirichlet process $P\sim DP(a,H)$ can be approximated by
\begin{equation}\label{approx of DP}
P_{N}=\sum_{i=1}^{N}J_{i,N}\delta_{Y_{i}},
\end{equation}
with the
monotonically decreasing weights
$J_{i,N}=\frac{G_{a/N}^{-1}(\frac{\Gamma_{i}}{\Gamma_{N+1}})}{\sum_{j=1}^{N}G_{a/N}^{-1}(\frac{\Gamma_{i}}{\Gamma_{N+1}})},$
where $\Gamma_{i}$ and $Y_{i}$ are defined as before, $N$ is a positive large
integer and $G_{a/N}$ denotes the complement-cdf of the $\text{gamma}(a/N,1)$ distribution.
Note that, $G^{-1}_{a/N}(p)$ is the $(1-p)$-th quantile of the $\text{gamma}(a/N,1)$ distribution.
The following
algorithm describes how the approximation \eqref{approx of DP} can be used to generate
a sample from $DP(a,H)$.

\begin{algorithm}
\caption{Approximately generating a sample from $DP(a,H)$}
\begin{enumerate}
\item Fix a large positive integer $N$ and generate i.i.d. $Y_{i}\sim H$ for $i=1,\ldots,N$.

\item For $i=1,\ldots,N+1$, generate i.i.d. $E_{i}$ from the exponential distribution with rate 1, independent
of $\left(Y_{i}\right)_{1\leq i\leq N}$ and put $\Gamma_{i}=E_{1}+\cdots+E_{i}$.
\item Compute $G_{a/N}^{-1}\left(  {\Gamma_{i}}/{\Gamma_{N+1}}\right)  $ for
$i=1,\ldots,N$ and return $P_{N}.$
\end{enumerate}
\end{algorithm}
The Dirichlet process can also be obtained from the following finite mixture models developed by Ishwaran and Zarepour (2002). Let $P_{N}$ has the from given \eqref{approx of DP} with $(J_{i,N})_{1\leq i\leq N}\sim$ Dirichlet$(a/N,\ldots,a/N)$. Then $E_{P_{N}}(g)\rightarrow E_{P}(g)$ in distribution as $N\rightarrow\infty$, for any measurable function $g: \mathbb{R}\rightarrow \mathbb{R}$ with $\int_{\mathbb{R}}|g(x)|\, H(dx)<\infty$ and $P\sim DP(a,H)$. In particular, $(P_{N})_{N\geq 1}$ converges in distribution to $P$, where $P_{N}$ and $P$ are random values in the space $M_{1}(\mathbb{R})$ of probability measures on $\mathbb{R}$ endowed with the topology of weak convergence. To generate $(J_{i,N})_{1\leq i\leq N}$ put $J_{i,N}=G_{i,N}/\sum_{i=1}^{N}G_{i,N}$, where $(G_{i,N})_{1\leq i\leq N}$ is a sequence of i.i.d. gamma$(a/N, 1)$ random variables independent of $(Y_{i})_{1\leq i\leq N}$. This form of approximation leads to some results in Section 5.
\subsection{Relative Belief Inferences}

The relative belief ratio, developed by Evans (2015), becomes a widespread measure of statistical evidence. See, for example, the work of Al-Labadi and Evans (2018), Al-Labadi et al. (2017, 2018), and Al-Labadi et al. (2019a, 2019b) for implementation of the relative belief ratio on different stimulating univariate hypothesis testing problems. In details, let $\{f_{\theta}:\theta\in\Theta\}$ be a collection of densities on a sample space $\mathfrak{X}$ and let $\pi$ be a prior on the parameter space $\Theta$. Note that
the densities may represent discrete or continuous probability measures but they are
all with respect to the same support measure  $d\theta$. After
observing the data $x,$ the posterior distribution of $\theta$, denoted by $\pi(\theta\,|\,x)$, is a revised prior and is given by the
density $\pi(\theta\,|\,x)=\pi(\theta)f_{\theta}(x)/m(x)$, where $m(x)=\int
_{\Theta}\pi(\theta)f_{\theta}(x)\,d\theta$ is the prior predictive density of
$x.$  For a parameter of interest $\psi=\Psi(\theta),$ let $\Pi_{\Psi}$ be
the marginal prior probability measure and $\Pi_{\Psi}(\cdot|\,x)$ be
the marginal posterior probability measure. It is assumed that
$\Psi$ satisfies regularity conditions
so that the prior density $\pi_{\Psi}$ and the posterior density
$\pi_{\Psi}(\cdot\,|\,x)$ of $\psi$ exist with respect to some support measure on the range space for $\Psi$
. The relative belief ratio for a value
$\psi$ is then defined by $RB_{\Psi}(\psi\,|\,x)=\lim_{\delta\rightarrow0}%
\Pi_{\Psi}(N_{\delta}(\psi\,)|\,x)/\Pi_{\Psi}(N_{\delta}(\psi\,)),$ where
$N_{\delta}(\psi\,)$ is a sequence of neighborhoods of $\psi$ converging
nicely to $\psi$ as $\delta\rightarrow0$ (Evans, 2015). When $\pi_{\Psi}$ and  $\pi_{\Psi}(\cdot\,|\,x)$ are continuous at $\psi,$ the relative belief ratio is defined by
\begin{equation*}
RB_{\Psi}(\psi\,|\,x)=\pi_{\Psi}(\psi\,|\,x)/\pi_{\Psi}(\psi), \label{relbel}%
\end{equation*}
the ratio of the posterior density to the prior density at $\psi.$  Therefore,
$RB_{\Psi}(\psi\,|\,x)$ measures the change in the belief of $\psi$ being the true value from a \textit{priori} to a \textit{posteriori}.

Since $RB_{\Psi}(\psi\,|\,x)$ is a measure of the evidence that $\psi$ is the true value, if $RB_{\Psi}(\psi\,|\,x)$ $>1$, then the probability of the $\psi$ being the true value from a priori to a posteriori is increased, consequently there is evidence based on the data that $\psi$ is the true value. If $RB_{\Psi}(\psi\,|\,x)<1$, then the probability of the $\psi$ being the true value from a priori to a posteriori is decreased. Accordingly, there is evidence against based on the data that $\psi$ being the true value. For the case $RB_{\Psi}(\psi\,|\,x)=1$ there is no
evidence either way.

Obviously, $RB_{\Psi}(\psi_{0}\,|\,x)$ measures the evidence of the hypothesis $\mathcal{H}_{0}:\Psi(\theta)=\psi_{0}$. Large values of $RB_{\Psi}(\psi_{0}\,|\,x)=c$ provides
 strong evidence in favor of $\psi_{0}$. However, there may also exist other
values of $\psi$ that had even larger increases. Thus, it is also necessary, however, to calibrate whether this is strong or weak evidence for
or against $\mathcal{H}_{0}.$ A typical
calibration of $RB_{\Psi}(\psi_{0}\,|\,x)$  is given by the  \textit{strength}
\begin{equation}
\Pi_{\Psi}\left[RB_{\Psi}(\psi\,|\,x)\leq RB_{\Psi}(\psi_{0}\,|\,x)\,|\,x\right].
\label{strength}%
\end{equation}
The value in \eqref{strength} indicates that the posterior probability that the true value of $\psi$ has a relative
belief ratio no greater than that of the hypothesized value $\psi_{0}.$ Noticeably, (\ref{strength}) is not a p-value as it has a very different
interpretation. When $RB_{\Psi}(\psi_{0}\,|\,x)<1$, there is evidence
against $\psi_{0},$ then a small value of (\ref{strength}) indicates
 strong evidence against $\psi_{0}$. On the other hand, a large value for \eqref{strength}    indicates   weak evidence against $\psi_{0}$.
Similarly, when $RB_{\Psi}(\psi_{0}\,|\,x)>1$, there is  evidence in favor
of $\psi_{0},$ then a small value of (\ref{strength}) indicates  weak
evidence in favor of $\psi_{0}$, while a large value of \eqref{strength} indicates
 strong evidence in favor of $\psi_{0}$.
\subsection{Energy Distance}
The \emph{Energy distance}, presented by Sz\'{e}kely (2003), is an appropriate tool to determine the equality of distributions. In general, the Energy distance between two $m$-variate distribution function $F$ and $G$ is defined by
\begin{equation}\label{E-distance-1}
d_{\mathcal{E}}(F,G)=2E\|\mathbf{X}-\mathbf{Y}\|-E\|\mathbf{X}-\mathbf{X}^{\prime}\|-E\|\mathbf{Y}-\mathbf{Y}^{\prime}\|,
\end{equation}
where $\mathbf{X},\mathbf{X}^{\prime}\overset{i.i.d.}{\sim} F$, $\mathbf{Y},\mathbf{Y}^{\prime}\overset{i.i.d.}{\sim} G$ and $\|\mathbf{a}\|=\sqrt{\mathbf{a}^{T}\mathbf{a}}$ denotes  Euclidean norm of vector $\mathbf{a}=(a_{1},\ldots,a_{m})$. Sz\'{e}kely and Rizzo (2013) showed that $d_{\mathcal{E}}(F,G)\geq 0$ such that equality holds if and only if $F=G$. Note that, from (Sz\'{e}kely, 2003), the Energy distance \eqref{E-distance-1} is rotation invariant. This property makes it appropriate for testing goodness-of-fit problems in higher dimensions. Specifically, let  $G$ be the hypothesized distribution and $\mathbf{x}_{m\times n}=(\mathbf{x}_{1},\ldots,\mathbf{x}_{n})$ be the observed sample from $F$. Then, the one sample Energy distance corresponding to \eqref{E-distance-1} is defined by
\begin{equation}\label{E-distance-2}
d_{\mathcal{E},n}(F,G)=\frac{2}{n}\sum_{i=1}^{n}E||\mathbf{x}_{i}-\mathbf{Y}||-\frac{1}{n^{2}}\sum_{\ell,m=1}^{n}||\mathbf{x}_{\ell}-\mathbf{x}_{m}||-E||\mathbf{Y}-\mathbf{Y}^{\prime}||,
\end{equation}
where $\mathbf{x}_{i}\in\mathbb{R}^{m}$, for $i=1,\ldots,n$, and the expectations are taken with respect to the distribution $G$. The special important case occurs when $G$ is a multivariate normal distribution where the $\mathsf{R}$ package \textbf{energy} is usually used for implementing \eqref{E-distance-2}. For further discussion about Energy distance consult Sz\'{e}kely (2003) and Sz\'{e}kely and Rizzo (2013).

\section{A Bayesian Semiparametric Gaussian Copula Approach for Modeling Multivariate Distributions}
In this section, we propose a Bayesian semiparametric copula approach based on the Gaussian copula as a flexible model for modeling multivariate distributions. For a brevity, we refer to this procedure as the BSPGC (Bayesian semiparametric Gaussian copula) approach. Specifically, let $\mathbf{x}_{m\times n}=(\mathbf{x}_{1},\ldots,\mathbf{x}_{n})$ be a sample of size $n$ from an unknown $m$-variate distribution $F_{true}$ with maginal cdf's $F_{1},\ldots,F_{m}$. Note that, the subscript $m\times n$ may be omitted whenever it is clear in the context.
To model $F_{true}$ based on the BSPGC approach we use the prior $F_{i}\sim DP(a,H_{i})$, where $H_{i}$ is the $i$-th marginal cdf of a given $m$-variate distribution $H$. So, by \eqref{pos base measure}, for a given choice of $a$, $F^{\ast}_{i}=F_{i}|\mathbf{x}_{i}\sim DP(a+n, H_{i}^{\ast})$ for $i=1,\ldots,m$.  Consider the joint cdf $H^{\ast}$ corresponding to marginal cdf's $H_{1}^{\ast},\ldots,H_{m}^{\ast}$ with correlation matrix $R^{\ast}$. Then, the \emph{posterior-based model} is defined by
\begin{equation}\label{pos-BSPGC}
F^{\ast}(t_{1},\ldots,t_{m})=C_{R^{\ast}}(F^{\ast}_{1}(t_{1}),\ldots, F^{\ast}_{m}(t_{m})).
\end{equation}
The next lemma shows that $F^{\ast}$ approaches to the true distribution $F_{true}$ when the sample size increases.
\begin{lemma}\label{post as n to inf}
Let $\mathbf{x}_{m\times n}$ be a sample from $m$-variate distribution function $F_{true}$ with unknown marginal cdf's $F_{1},\ldots,F_{m}$. Assume that $F^{\ast}_{i}\sim DP(a+n, H_{i}^{\ast})$, for $ i=1,\ldots,m$. For any $\mathbf{t}=(t_{1},\ldots,t_{m})\in \mathbb{R}^{m}$, $C_{R^{\ast}}(F^{\ast}_{1}(t_{1}),\ldots, F^{\ast}_{m}(t_{m}))\xrightarrow{a.s.} F_{true}(\mathbf{t})$ as $n\rightarrow\infty$.
\end{lemma}
\begin{proof}
For any $\mathbf{t}\in\mathbb{R}^{m}$, the triangle inequality implies
\begingroup\makeatletter\def\f@size{9}\check@mathfonts
  \begin{align}\label{triangular}
\big|C_{R^{\ast}}(F^{\ast}_{1}(t_{1}),\ldots, F^{\ast}_{m}(t_{m}))-F_{true}(\mathbf{t})\big|&\leq\big|C_{R^{\ast}}(F^{\ast}_{1}(t_{1}),\ldots, F^{\ast}_{m}(t_{m}))-H^{\ast}(\mathbf{t})\big|\nonumber\\
&+ \big|H^{\ast}(\mathbf{t})-F_{true}(\mathbf{t})\big|=I_{1}+I_{2}.
\end{align}
\endgroup
From \eqref{pos base measure}, for any $\mathbf{t}\in\mathbb{R}^{m}$, $H^{\ast}(\mathbf{t})\xrightarrow{a.s.} F_{true}(\mathbf{t})$ as $n\rightarrow\infty$. Hence, the continuous mapping theorem implies $I_{2}\xrightarrow{a.s.} 0$ as $n\rightarrow\infty$. On the other hand, from Sklar's theorem \ref{sklar's Thm} and Lipschitz condition \eqref{Lipschitz}, we have
\begin{small}
\begin{align*}
I_{1}&=
\big|C_{R^{\ast}}(F^{\ast}_{1}(t_{1}),\ldots, F^{\ast}_{m}(t_{m}))-C_{R^{\ast}}(H^{\ast}_{i}(t_{1}),\ldots, H^{\ast}_{i}(t_{m}))\big|\\&\leq\sum_{i=1}^{m}\big| (F^{\ast}_{i}(t_{i})-H^{\ast}_{i}(t_{i})\big|.
\end{align*}
\end{small}
Note that, from the property of the Dirichlet process, for any $t_{i}\in\mathbb{R}$ and $\epsilon>0$, Chebyshev's inequality implies
\begin{align*}
Pr\left\lbrace \big| F^{\ast}_{i}(t_{i})-H^{\ast}_{i}(t_{i})\big|\geq\epsilon\right\rbrace
\leq\dfrac{H^{\ast}_{i}(t_{i})\left(1-H^{\ast}_{i}(t_{i})\right)}{(a+n+1)\epsilon^{2}}.
\end{align*}
Substituting $n=k^{2}+b$, for $k\in\mathbb{N}$ and $b\in\lbrace 0,1,\ldots\rbrace$, gives
\begin{align*}
Pr\left\lbrace \big| F^{\ast}_{i}(t_{i})-H^{\ast}_{i}(t_{i})\big|\geq\epsilon\right\rbrace
\leq\dfrac{1}{4\,k^{2}\epsilon^{2}}.
\end{align*}
Since $\sum_{k=1}^{\infty}k^{-2}$ converges, then $\sum_{k=1}^{\infty}Pr\left\lbrace \big| F^{\ast}_{i}(t_{i})-H^{\ast}_{i}(t_{i})\big|\geq\epsilon\right\rbrace<\infty.$ Hence, by the first Borel Cantelli lemma, $\big| F^{\ast}_{i}(t_{i})-H^{\ast}_{i}(t_{i})\big|\xrightarrow{a.s.}0$, as $k\rightarrow\infty$ or $n\rightarrow\infty$. This completes the proof.
\end{proof}
\section{Selecting $a$, $H$ and the Method of Estimation of $R^{\ast}$ in the BSPGC Approach}
The proposed method for modeling multivariate distributions depends on $a$, $H$ and $R^{\ast}$. Hence, it is necessary to look carefully at the impact of these parameters on the approach. For instance, from \eqref{pos base measure}, a large value of $a$ can increase the effect of the $m$-variate distribution $H$ instead of $F_{n}$ in the posterior-based model \eqref{pos-BSPGC}. The following lemma shows the effect of the value of $a$ on the model \eqref{pos-BSPGC}.
\begin{lemma}\label{effect a-pos}
Let $\mathbf{x}_{m\times n}$ be a sample from $m$-variate distribution function $F_{true}$ with unknown marginal cdf's $F_{1},\ldots,F_{m}$. Also, let $H$ be a known $m$-variate cdf with marginal cdf's $H_{1},\ldots,H_{m}$ and $H^{\ast}$ be the $m$-variate cdf, defined in \eqref{pos base measure}, with marginal cdf's $H^{\ast}_{1},\ldots,H^{\ast}_{m}$. Assume that $F^{\ast}_{i}\sim DP(a+n, H^{\ast}_{i})$, for $1\leq i\leq m$. Then, for any $\mathbf{t}\in \mathbb{R}^{m}$, $C_{R^{\ast}}(F^{\ast}_{1}(t_{1}),\ldots, F^{\ast}_{m}(t_{m}))\xrightarrow{a.s.}H(\mathbf{t})$ as $a\rightarrow\infty$.
\end{lemma}
\begin{proof}
The proof is similar to the proof of Lemma \ref{post as n to inf}. For this, assume that $a=k^{2}c$ for $k\in\lbrace 0,1,\ldots\rbrace$ and a fixed positive number $c$. For any fixed $n$, replace $F_{true}(\mathbf{t})$ and $a$, respectively by, $H(\mathbf{t})$ and $k^{2}c$ in the proof of Lemma \ref{post as n to inf}. Then the result follows.
\end{proof}\\

It follows from Lemma \ref{effect a-pos} that increasing the value of $a$ can lead to some errors. To avoid this issue, we propose to choose $a$ to be at most $0.5\,n$ as recommended in Al-Labadi and Zarepur (2017).

The choice of $H$ is also very significant and there are two main issues to reflect. The first one is the independence of the approach to the choice of $H$ (invariance property).   As pointed in Table \ref{priori-invariant}, the approach is invariance to the choice of any continuous multivariate distribution. The second issue is the compatibility between $H$ and the data. This is typically called \emph{prior-data conflict} (Evans and Moshonov, 2006; Al-Labadi and Evans, 2017,  Al-Labadi and Evans, 2018, Al-Labadi and Wang, 2019). As illustrated in Section 7.2, the existence of prior-data conflict yields to a failure of the approach and thus should be avoided. Since $E(F_{i})=H_{i}$, where $F_{i}\sim DP(a,H_{i})$ for $i=1,\ldots,m$, a reasonable choice of $H$ that ensures the avoidance of prior-data conflict  is the $m$-variate normal distribution $N_{m}(\overline{\mathbf{x}},S_{\mathbf{x}})$, where $\overline{\mathbf{x}}=1/n\sum_{i=1}^{n}\mathbf{x}_{i}$ and $S_{\mathbf{x}}=1/(n-1)\sum_{i=1}^{n}(\mathbf{x}_{i}-\overline{\mathbf{x}})(\mathbf{x}_{i}-\overline{\mathbf{x}})^{T}$.

To carry on the approach, it is essential to estimate the correlation matrix $R^{\ast}$. For this, we first generate a sample from the mixture distribution in \eqref{pos base measure}. Then, based on the generated sample, $R^{\ast}$  is estimated by one of the following three common procedures: the Gaussian correlation rank, the Kendal's $\tau$ or the Spearman's $\rho$. In Section 7, we performed a simulation study to compare the effect of these three methods on the quality of the approach. As a result,  we recommend using Kendal's $\tau$ correlation coefficients with $a=1$ in the proposed model.

\section{A MVN Test Based on the BSPGC Approach}
Let $\mathbf{x}_{m\times n}$ be a sample of size $n$ from an unknown $m$-variate  distribution $F_{true}$. The  problem to be addressed in this section is to test the hypothesis
\begin{equation}\label{test1}
\mathcal{H}_{0}:F_{true}\in \mathcal{F},
\end{equation}
\noindent where $\mathcal{F}=\{N_{m}(\boldsymbol{\mu}_{m},\Sigma_{m}):\boldsymbol{\mu}_{m}\in\mathbb{R}^{m},\det (\Sigma_{m}) >0\}$. Note that, whenever $\boldsymbol{\mu}_{m}$ and $\Sigma_{m}$ are unknown, they are to be estimated by the sample mean vector $\overline{\mathbf{x}}$ and sample covariance matrix
$S_{\mathbf{x}}$, respectively. Thus, for $\boldsymbol{\theta}_{\mathbf{x}}=(\overline{\mathbf{x}},S_{\mathbf{x}})$, $F_{\boldsymbol{\theta}_{\mathbf{x}}}=N_{m}(\overline{\mathbf{x}},S_{\mathbf{x}})$
is the best representative of the family $\mathcal{F}$ to compare with distribution $F_{true}$. Hence, testing \eqref{test1} is equivalent to test
\begin{equation}\label{test2}
\mathcal{H}_{0}:F_{true}=F_{\boldsymbol{\theta}_{\mathbf{x}}}.
\end{equation}
Now,  we continue as follows. 
Let $H=F_{\boldsymbol{\theta}_{\mathbf{x}}}$ with marginal cdf's $H_{1}=F_{\boldsymbol{\theta}_{\mathbf{x}_{1}}},\ldots,H_{m}=F_{\boldsymbol{\theta}_{\mathbf{x}_{m}}}$. Here, for $i=1,\ldots,m$, $F_{\boldsymbol{\theta}_{\mathbf{x}_{i}}}$ is the cdf of the univariate normal distribution with mean $\overline{x}_{i}$ and variance $s^{2}_{i}$, where $\overline{x}_{i}$ and $s^{2}_{i}$ are the $i$-th element of $\overline{\mathbf{x}}$ and $i$-th diagonal element of $S_{\mathbf{x}}$, respectively. Assume that $F_{i}\sim DP(a, H_{i})$. For any $\mathbf{t}\in\mathbb{R}^{m}$, let
\begin{equation}\label{pri-BSPGC}
F(t_{1},\ldots,t_{m})=C_{R_{\mathbf{x}}}(F_{1}(t_{1}),\ldots, F_{m}(t_{m})),
\end{equation}
be the \emph{prior-based model}, where $R_{\mathbf{x}}=\left(r_{\mathbf{x},ij}\right)$ is the correlation matrix of the $m$-variate distribution $F_{\boldsymbol{\theta}_{\mathbf{x}}}$ and to be estimated as discussed in Section 4. Note that, as pointed out earlier, setting  $H_{i}=F_{\boldsymbol{\theta}_{\mathbf{x}_{i}}}$ ensures compatibility between the data and the prior which will certainly avoid prior-data conflict. More details about the effect of the prior-data conflict on the approach is clarified in Section 7, where it is revealed that  the existence of prior data conflict  leads to erroneous result of the test.

Recalling  the posterior-based model as defined in Section 3, to proceed  with the approach, the energy distance is used to compute the distance between this model and $F_{\boldsymbol{\theta}_{\mathbf{x}}}$ (posterior distance) and the distance between the prior-based model and $F_{\boldsymbol{\theta}_{\mathbf{x}}}$ (prior distance). The next lemma proposes a Bayesian counterpart of the distance \eqref{E-distance-2} as an appropriate tool to measure dissimilarities between the proposed models and the null distribution. This is considered a very convenient tool for assessing MVN in high dimensional problems ($m>n$).
\begin{lemma}\label{Bayesian-lemma}
Let $F_{true}$ be an $m$-variate distribution function with unknown marginal cdf's $F_{1},\ldots,F_{m}$ and $H$ be a known $m$-variate distribution function with marginal cdf's $H_{1},\ldots,H_{m}$. Assume that $F_{i}\sim DP(a,H_{i})$, and $F_{N_i}$'s are the approximation of Dirichlet process $F_{i}$'s, given by Ishwaran and Zarepour (2002). Consider the Energy distance between $F_{N}=C_{R}(F_{N_1}(t_{1}),\ldots,F_{N_m}(t_{m}))$ and $H$ as
\begin{equation}\label{Beyesian-E-distance}
d_{\mathcal{E},N,a}(F_{N},H)=2\sum_{i=1}^{N}J_{i,N}E_{H}||\mathbf{x}_{i}-\mathbf{Y}||-\sum_{i,j=1}^{N}J_{i,N}J_{j,N}||\mathbf{x}_{i}-\mathbf{x}_{j}||-E_{H}||\mathbf{Y}-\mathbf{Y}^{\prime}||,
\end{equation}
where $(J_{i,N})_{1\leq i\leq N}\sim Dirichlet(a/N,\ldots,a/N)$, $\mathbf{Y},\mathbf{Y}^{\prime}\overset{i.i.d.}{\sim}H$, $(\mathbf{x}_{1},\ldots,\mathbf{x}_{N})$ is an observed sample from $F_{N}$ and $R$ is the correlation matrix of $H$. Then, as $a\rightarrow\infty$
\begin{equation*}
E_{F_{N}}(d_{\mathcal{E},N,a}(F_{N},H))\xrightarrow{a.s.}d_{\mathcal{E},N}(F_{N},H),
\end{equation*}
where $d_{\mathcal{E},N}(F_{N},H)$ is defined in \eqref{E-distance-2} with $F=F_{n}$, $G=H$ and $n=N.$
\end{lemma}
\begin{proof}
From properties of Dirichlet distribution, $E_{F_{N}}(J_{i,N})=1/N$ and $E_{F_{N}}(J_{i,N}J_{j,N})$ $=a/((a+1)N^{2})$. Then
\begin{equation}\label{mean-E-dist}
E_{F_{N}}(d_{\mathcal{E},N,a}(F_{N},H))=\frac{2}{N}\sum_{i=1}^{N}E_{H}||\mathbf{x}_{i}-\mathbf{Y}||-\frac{a}{(a+1)N^{2}}\sum_{i,j=1}^{N}||\mathbf{x}_{i}-\mathbf{x}_{j}||-E_{H}||\mathbf{Y}-\mathbf{Y}^{\prime}||.
\end{equation}
The proof is immediately followed by letting $a\rightarrow\infty$ in \eqref{mean-E-dist}.
\end{proof}

The next lemma allows us to use the approximation of the Dirichlet process in the prior-based and posterior-based models for approximating the distribution of the posterior and the prior distances computed by \eqref{Beyesian-E-distance}.
\begin{lemma}
Let $F_{true}$ be an $m$-variate distribution function with unknown marginal cdf's $F_{1},\ldots,F_{m}$ and $H$ be a known $m$-variate distribution function with marginal cdf's $H_{1},\ldots,H_{m}$. Assume that $F_{i}\sim DP(a,H_{i})$ and $F_{Ni}$'s are the approximation of the Dirichlet process $F_{i}$'s, given in \eqref{approx of DP}. Then, for any $(t_{1},\ldots,t_{m})\in\mathbb{R}^{m}$, $C_{R}(F_{N1}(t_{1}),\ldots,$ $F_{Nm}(t_{m}))\xrightarrow{a.s.}C_{R}(F_{1}(t_{1}),\ldots, F_{m}(t_{m}))$, as $N\rightarrow\infty$, where $R$ is the correlation matrix of $H$.
\end{lemma}

\begin{proof}
From Lipschitz condition \eqref{Lipschitz}, we have
\begin{small}
\begin{align*}
\big|C_{R}(F_{N1}(x_{1}),\ldots, F_{Nm}(x_{m}))-C_{R}(F_{1}(x_{1}),\ldots, F_{m}(x_{m}))\big|\leq\sum_{i=1}^{m}\big| F_{N1}(x_{i})-F_{1}(x_{i})\big|.
\end{align*}
\end{small}
Since $F_{Ni}(x_{i})\xrightarrow{a.s.}F_{i}(x_{i})$, for $1\leq i\leq m$ (Zarepur and Al-Labadi, 2012), the result follows.
\end{proof}

The procedure is continued by considering $d_{\mathcal{E},N,a}(F_{N},F_{\boldsymbol{\theta}_{\mathbf{x}}})$ as the Energy distance between the prior-based model \eqref{pri-BSPGC} and the null distribution $F_{\boldsymbol{\theta}_{\mathbf{x}}}$ using formula \eqref{Beyesian-E-distance}. Similarly, consider $d_{\mathcal{E},N,a}(F^{\ast}_{N},F_{\boldsymbol{\theta}_{\mathbf{x}}})$ for the posterior-based model \eqref{pos-BSPGC} and the null distribution. Then, the relative belief ratio is used to compare the concentration of the posterior distribution  $d_{\mathcal{E},N,a}(F^{\ast}_{N},F_{\boldsymbol{\theta}_{\mathbf{x}}})$ to the prior distribution $d_{\mathcal{E},N,a}(F_{N},F_{\boldsymbol{\theta}_{\mathbf{x}}})$ about zero. As shown in the next lemma, if $\mathcal{H}_{0}$ is true, the distribution of the posterior distance should be more concentrated about 0 than the distribution of the prior distance; otherwise, the distribution of the prior distance should be more concentrated at 0 than the distribution of posterior distance. The comparison is made by computing the relative belief ratio with the interpretation as discussed
in$\ $Section 2.
\begin{lemma}\label{lem-2}
Let $\mathbf{x}_{m\times n}$ be a sample from $m$-variate distribution function $F_{true}$ with unknown marginal cdf's $F_{1},\ldots,F_{m}$. Assume that $\boldsymbol{\theta}_{\mathbf{x}}\xrightarrow{a.s.}\boldsymbol{\theta}_{0}$ and $F_{\boldsymbol{\theta}_{\mathbf{x}}}\xrightarrow{a.s.} F_{\boldsymbol{\theta}_{0}}$ as $n\rightarrow\infty$. Let $F^{\ast}_{i}\sim DP(a+n, H^{\ast}_{i})$, for $i=1,\ldots,m$. For any $\mathbf{t}=(t_{1},\ldots,t_{m})\in \mathbb{R}^{m}$ as $n\rightarrow\infty$

\begin{enumerate}[nolistsep,label=(\roman*),leftmargin=\parindent,align=left,labelwidth=\parindent,labelsep=0pt]
\item
$
\big|C_{R^{\ast}}(F^{\ast}_{1}(t_{1}),\ldots, F^{\ast}_{m}(t_{m}))-F_{\boldsymbol{\theta}_{\mathbf{x}}}(t_{1},\ldots,t_{m})\big|\xrightarrow{a.s.}0,
$ when $\mathcal{H}_{0}$ is true.
\item
$
\liminf\big|C_{R^{\ast}}(F^{\ast}_{1}(t_{1}),\ldots, F^{\ast}_{m}(t_{m}))-F_{\boldsymbol{\theta}_{\mathbf{x}}}(t_{1},\ldots,t_{m})\big|\displaystyle{\overset{a.s.}{>}} 0,
$ when $\mathcal{H}_{0}$ is not true,
\end{enumerate}
where $R^{\ast}$ is the correlation matrix of $H^{\ast}$, defined in \eqref{pos base measure}.
\end{lemma}

\begin{proof}
To prove (i), substitute $F_{true}$ in \eqref{triangular} by $F_{\boldsymbol{\theta}_{\mathbf{x}}}$. From \eqref{pos base measure}, for any $\mathbf{t}\in\mathbb{R}^{m}$, $H^{\ast}(\mathbf{t})\xrightarrow{a.s.} F_{true}(\mathbf{t})$ as $n\rightarrow\infty$. If $\mathcal{H}_{0}$ is true, then $F_{true}(\mathbf{t})=F_{\boldsymbol{\theta}_{0}}(\mathbf{t})$. Hence, the proof of (i) immediately follows from the proof of Lemma \ref{post as n to inf}. To prove (ii), Consider $I_{1}$ as in \eqref{triangular}. Applying the triangle inequality gives
\begin{align*}
\big|C_{R^{\ast}}(F^{\ast}_{1}(t_{1}),\ldots, F^{\ast}_{m}(t_{m}))-F_{\boldsymbol{\theta}_{\mathbf{x}}}(t_{1},\ldots,t_{m})\big|\geq \big|H^{\ast}(\mathbf{t})- F_{\boldsymbol{\theta}_{\mathbf{x}}}(\mathbf{t})\big|-I_{1}.
\end{align*}
Similar to the proof of Lemma \ref{post as n to inf}, $I_{1}\xrightarrow{a.s.}0$ and $\big|H^{\ast}(\mathbf{t})- F_{\boldsymbol{\theta}_{\mathbf{x}}}(\mathbf{t})\big|\xrightarrow{a.s.}\big|F_{true}(\mathbf{t})- F_{\boldsymbol{\theta}_{0}}(\mathbf{t})\big|$ as $n\rightarrow\infty$. Since $\mathcal{H}_{0}$ is not true, $\big|F_{true}(\mathbf{t})- F_{\boldsymbol{\theta}_{0}}(\mathbf{t})\big|\displaystyle{\overset{a.s.}{>}} 0$ which completes the proof of (ii).
\end{proof}

The effect of the value of $a$ on the posterior-based model was considered in Lemma \ref{effect a-pos}. It is also interesting to consider the effect of the value of $a$ on the proposed MVN test.
\begin{lemma}\label{effect a-pri}
Let $\mathbf{x}_{m\times n}$ be a sample from $m$-variate distribution function $F_{true}$ with unknown marginal cdf's $F_{1},\ldots,F_{m}$. Let $F_{\boldsymbol{\theta}_{\mathbf{x}}}$ be the cdf of $N_{m}(\overline{\mathbf{x}},S_{\mathbf{x}})$ with marginal cdf's $F_{\boldsymbol{\theta}_{\mathbf{x}_{1}}},\ldots,F_{\boldsymbol{\theta}_{\mathbf{x}_{m}}}$. If $F_{i}\sim DP(a,F_{\boldsymbol{\theta}_{\mathbf{x}_{i}}})$, for $i=1,\ldots,m$, then $C_{R_{\mathbf{x}}}(F_{1}(t_{1}),\ldots, F_{m}(t_{m}))$ $\xrightarrow{a.s.}F_{\boldsymbol{\theta}_{\mathbf{x}}}(\mathbf{t})$ as $a\rightarrow\infty$, where $R_{\mathbf{x}}$ is the correlation matrix of $F_{\boldsymbol{\theta}_{\mathbf{x}}}$.
\end{lemma}

\begin{proof}
The proof is similar to the proof of Lemma \ref{effect a-pos} and is omitted.
\end{proof}

Lemmas \ref{effect a-pos} and \ref{effect a-pri} show that for a too large value of $a$  (relative to the sample size) both the posterior-based and prior-based models are approaching to the null model $F_{\boldsymbol{\theta}_{\mathbf{x}}}$. Hence, the comparison between the posterior and prior distance to detect the normality can lead to an error in which we may accept $\mathcal{H}_{0}$ when it is not true and reject $\mathcal{H}_{0}$ when it is true. As recommended in Section 7, we should consider $a$ at most $0.5\,n$.

At the end of this section, it is worth pointing out that the proposed test can be extended to assess any family of multivariate distributions. For this, it is enough to consider a different family of multivariate distributions in the hypothesis \eqref{test1} and use its best representative distribution as $H$ in the methodology, which may be more challenging for some multivariate models.

\section{Main Steps for Testing the MVN}
The following computational algorithm summaries the main steps to test $\mathcal{H}_{0}$. This algorithm is viewed as a generalized version of Algorithm B of Al-Labadi and Evans (2018). Observe that, since closed forms of the densities of $D_{\mathcal{E}}=d_{\mathcal{E},N,a}(F_{N},F_{\boldsymbol{\theta}_{\mathbf{x}}})$ and  $D_{\mathcal{E}}|\mathbf{x}=d_{\mathcal{E},N,a}(F^{\ast}_{N},F_{\boldsymbol{\theta}_{\mathbf{x}}})$ are typically not available,   relative belief ratios need to be approximated via simulation. 


\noindent\rule[0.1ex]{\linewidth}{1pt}\\[-.1cm]
\textbf{Algorithm 3} \vspace{-.05cm}{Relative belief algorithm based on the BSPGC approach for testing MVN}
\\[-.08cm]\noindent\rule[0.5ex]{\linewidth}{.4pt}
\begin{small}
\begin{enumerate}
\smallskip
\item Use Algorithm 2 to generate (approximately) marginal cdf $F_{i}$'s from $ DP(a, F_{\boldsymbol{\theta}_{\mathbf{x}_{i}}})$, for $1\leq i\leq m$.

\item Generate a sample of $N$ values from the $m$-variate distribution $F_{\boldsymbol{\theta}_{\mathbf{x}}}$ and estimate the correlation matrix $R_{\mathbf{x}}$, denoted by $\widehat{R}_{\mathbf{x}}$, as discussed in  Section 4.

\item Use the generated marginal cdf's $F_{i}$ and set $R=\widehat{R}_{\mathbf{x}}$ in Algorithm 1 to get a sample of $N$ values from prior-based model \eqref{pri-BSPGC}.

\item Use \eqref{Beyesian-E-distance} for the sample generated  in steps 3 to compute the prior distance $d_{\mathcal{E},N,a}(F_{N},F_{\boldsymbol{\theta}_{\mathbf{x}}})$.

\item Repeat steps (1)-(4) to obtain a sample of $r$ values from the
prior of $D_{\mathcal{E}}$.

\item Repeat steps (1)-(5) by replacing $a$ by $a+n$, $F_{i}$ by $F^{\ast}_{i}$,  $F_{\boldsymbol{\theta}_{\mathbf{x}_{i}}}$ by $H^{\ast}_{i}$, $F_{\boldsymbol{\theta}_{\mathbf{x}}}$  by $H^{\ast}$, $R_{\mathbf{x}}$ by $R^{\ast}$, $F_{N}$ by $F^{\ast}_{N}$, $D_{\mathcal{E}}$ by $D_{\mathcal{E}}|\mathbf{x}$ and prior by posterior. This yields to a sample of $r$ values from the
posterior of $D_{\mathcal{E}}|\mathbf{x}$.

%
%
%
%
%

\item Let $M$ be a positive number. Let $\hat{F}_{D_{\mathcal{E}}}$ denote the
empirical cdf of $D_{\mathcal{E}}$ based on the prior sample in step (5) and for $i=0,\ldots,M,$
let $\hat{d}_{i/M}$ be the estimate of $d_{i/M},$ the $(i/M)$-th prior
quantile of $D_{\mathcal{E}}.$ Here $\hat{d}_{0}=0$, and $\hat{d}_{1}$ is the largest value
of $d_{\mathcal{E}}$. Let $\hat{F}_{D_{\mathcal{E}}}(\cdot\,|\mathbf{x})$ denote the empirical cdf of $D_{\mathcal{E}}|\mathbf{x}$ based
on the posterior sample in step (10). For $d\in\lbrack\hat{d}_{i/M},\hat
{d}_{(i+1)/M})$, estimate $RB_{D_{\mathcal{E}}}(d\,|\,\mathbf{x})={\pi_{D_{\mathcal{E}}}(d|\mathbf{x})}/{\pi_{D_{\mathcal{E}}}(d)}$ by
\begin{equation}
\widehat{RB}_{D_{\mathcal{E}}}(d\,|\,\mathbf{x})=M\{\hat{F}_{D_{\mathcal{E}}}(\hat{d}_{(i+1)/M}\,|\,\mathbf{x})-\hat{F}%
_{D_{\mathcal{E}}}(\hat{d}_{i/M}\,|\,\mathbf{x})\}, \label{rbest}%
\end{equation}
the ratio of the estimates of the posterior and prior contents of $[\hat
{d}_{i/M},\hat{d}_{(i+1)/M}).$ Thus, we estimate $RB_{D_{\mathcal{E}}}(0\,|\,\mathbf{x})=\frac{\pi_{D_{\mathcal{E}}}(0|\mathbf{x})}{\pi_{D_{\mathcal{E}}}(0)}$
 by $\widehat{RB}_{D_{\mathcal{E}}}(0\,|\,\mathbf{x})=M\widehat{F}_{D_{\mathcal{E}}}(\hat{d}_{p_{0}}\,|\,\mathbf{x})$ where
$p_{0}=i_{0}/M$ and $i_{0}$ is chosen so that $i_{0}/M$ is not too small
(typically $i_{0}/M\approx0.05)$.\textbf{\smallskip}
\item Estimate the strength $DP_{D_{\mathcal{E}}}\big(RB_{D_{\mathcal{E}}}(d\,|\,\mathbf{x})\leq RB_{D_{\mathcal{E}}}%
(0\,|\,\mathbf{x})\,|\,\mathbf{x}\big)$ by the finite sum
\begin{equation}
\sum_{\{i\geq i_{0}:\widehat{RB}_{D}(\hat{d}_{i/M}\,|\,\mathbf{x})\leq\widehat{RB}%
_{D}(0\,|\,\mathbf{x})\}}\big(\hat{F}_{D_{\mathcal{E}}}(\hat{d}_{(i+1)/M}\,|\,\mathbf{x})-\hat{F}_{D_{\mathcal{E}}}(\hat
{d}_{i/M}\,|\,\mathbf{x})\big). \label{strest}%
\end{equation}
\end{enumerate}
\end{small}
\noindent\rule[0.5ex]{\linewidth}{.4pt}

For fixed $M,$ as $N\rightarrow\infty,r\rightarrow\infty,$
then $\hat{d}_{i/M}$ converges almost surely to $d_{i/M}$ and (\ref{rbest})
and (\ref{strest}) converge almost surely to $RB_{D_{\mathcal{E}}}(d\,|\,\mathbf{x})$ and
$DP_{D_{\mathcal{E}}}\big(RB_{D_{\mathcal{E}}}(d\,|\,\mathbf{x})\leq RB_{D_{\mathcal{E}}}(0\,|\,\mathbf{x})\,|\,\mathbf{x}\big)$, respectively.  The consistency of the proposed test is achieved by Proposition 6 of Al-Labadi and Evans (2018).  As a recommendation, one should try different values of $a$ to make sure the right conclusion has been obtained. However, we found out that setting $a=1$ gives adequate results. More details about implementing the approach is discussed in the following section.

\section{Simulation Studies}
This section is divided into two subsections. In the first subsection, the quality of the approach to model multivariate distributions is investigated, where different choices of $a$, $H$ and $R^{\ast}$ are considered. The evaluation technique relies on using the mean of the Energy distance $\overline{d}_{\mathcal{E},N}(F^{\ast},F_{true})$ based on $r$ replications. Note that, from Lemma \ref{Bayesian-lemma}, one may consider using  the package \textbf{energy} available in $\mathsf{R}$ to compute the distance.  We generated samples each of size $n=1000$ from a variety of bivariate distributions. The notations of the used distributions are listed (Table \ref{notation}) in Appendix A. In this study, we set $N=r=1000$ in Algorithm 3 with steps (1)-(6). Note that, for the methodology to work well, we expect $\overline{d}_{\mathcal{E},N}(F^{\ast},F_{true})$ to be close to zero. In the  second subsection, the proposed test is illustrated through several examples.

\subsection{Checking the Quality of the Posterior-based Model}
The performance of the posterior-based model (i.e. the quality of estimating the model) is illustrated by  considering the bivariate distributions given in Table \ref{performance} with some choices of $a$. The results are reported based on the Kendall's correlation coefficients. From Table \ref{performance}, the close values of $\overline{d}_{\mathcal{E},N}(F^{\ast},F_{true})$ to zero indicates to the good performance of the methodology to model bivariate distributions, particularly when  $a=1$. Note that, as mentioned in Lemma \ref{effect a-pos}, with increasing the value of $a$, the accuracy of the methodology will be decreased. For more illustration, part (a) of Figure \ref{box.a.norm} gives the boxplots of the energy distance between $F_{true}=N_{2}(\mathbf{0}_{2},A_{2})$ and its corresponding $F^{\ast}$ for $a=1,\,5,\,10.$ Boxplots of marginal distributions are also given in part (b) of Figure \ref{box.a.norm}. Also, the marginal densities of $N_{2}(\mathbf{0}_{2},A_{2})$ and its $F^{\ast}$ are given in Figure \ref{fig 1}. The bivariate scatter plots are shown below the diagonal, histograms on the diagonal and the Kendall correlation above the diagonal. Correlation ellipses and loess smooths (red lines) are also shown.
\setlength{\extrarowheight}{.1mm}
\begin{table}[ht]
\centering
\setlength{\tabcolsep}{2.1 mm}
\caption{The mean of the Energy distance between the true and the posterior-based model based on the Kendall's $\tau$ for $a=1,\,5,\,10$ and $n=1000$.}\label{performance}
\scalebox{0.86}{
\begin{tabular}{ccc||ccc}
\toprule
\bfseries True distribution & $a$ &\scalebox{1.2}{$\overline{d}_{\mathcal{E},N}(F^{\ast},F_{true})$}&\bfseries True distribution & $a$ &\scalebox{1.2}{$\overline{d}_{\mathcal{E},N}(F^{\ast},F_{true})$}\\\hline

&\scalebox{1.1}{1}& \scalebox{1.1}{0.00524} & & \scalebox{1.1}{1}&\scalebox{1.1}{0.00528} \\
\scalebox{1.1}{$N_{2}(\mathbf{0}_{2},A_{2})$} &\scalebox{1.1}{5}& \scalebox{1.1}{0.00559}  &\scalebox{1.1}{$(\mathcal{P}_{VII}(1,1,1))^{2}$} & \scalebox{1.1}{5} & \scalebox{1.1}{0.00551}\\
&\scalebox{1.1}{10}& \scalebox{1.1}{0.00581} & & \scalebox{1.1}{10} & \scalebox{1.1}{0.00573}\\
\hline
&\scalebox{1.1}{1}& \scalebox{1.1}{0.00517}  & & \scalebox{1.1}{1}&\scalebox{1.1}{0.00533}\\
\scalebox{1.1}{$t_{5}(\mathbf{0}_{2},I_{2})$} &\scalebox{1.1}{5}& \scalebox{1.1}{0.00538} &\scalebox{1.1}{$E(0.5)\otimes E(0.25)$} & \scalebox{1.1}{5} & \scalebox{1.1}{0.00575}  \\
&\scalebox{1.1}{10}& \scalebox{1.1}{0.00563} & & \scalebox{1.1}{10}&\scalebox{1.1}{0.00581} \\
\hline
&\scalebox{1.1}{1}& \scalebox{1.1}{0.00564} & & \scalebox{1.1}{1}& \scalebox{1.1}{0.00520} \\
\scalebox{1.1}{$LN_{2}(\mathbf{0}_{2},B_{2})$} &\scalebox{1.1}{5}& \scalebox{1.1}{0.00568} &\scalebox{1.1}{$B(1,2)\otimes B(2,1)$} & \scalebox{1.1}{5} & \scalebox{1.1}{0.00557} \\
&\scalebox{1.1}{10}& \scalebox{1.1}{0.00597} & & \scalebox{1.1}{10} & \scalebox{1.1}{0.00563}\\
\bottomrule
\end{tabular}
}
\end{table}

\begin{figure}[ht]
 \centering
    \subfloat[]{{\includegraphics[width=5.2cm]{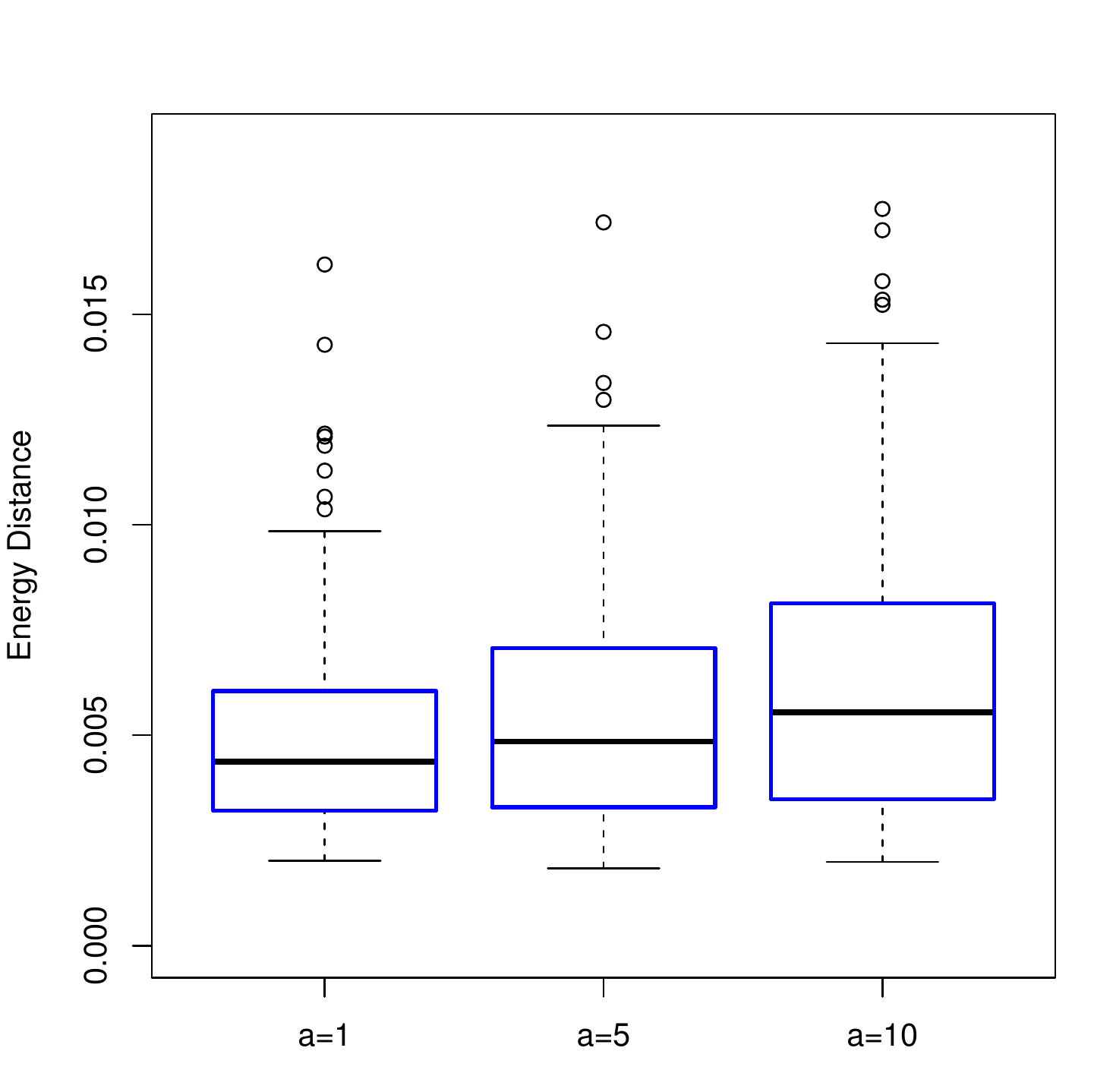} }}%
    \qquad
    \subfloat[]{{\includegraphics[width=5.2cm]{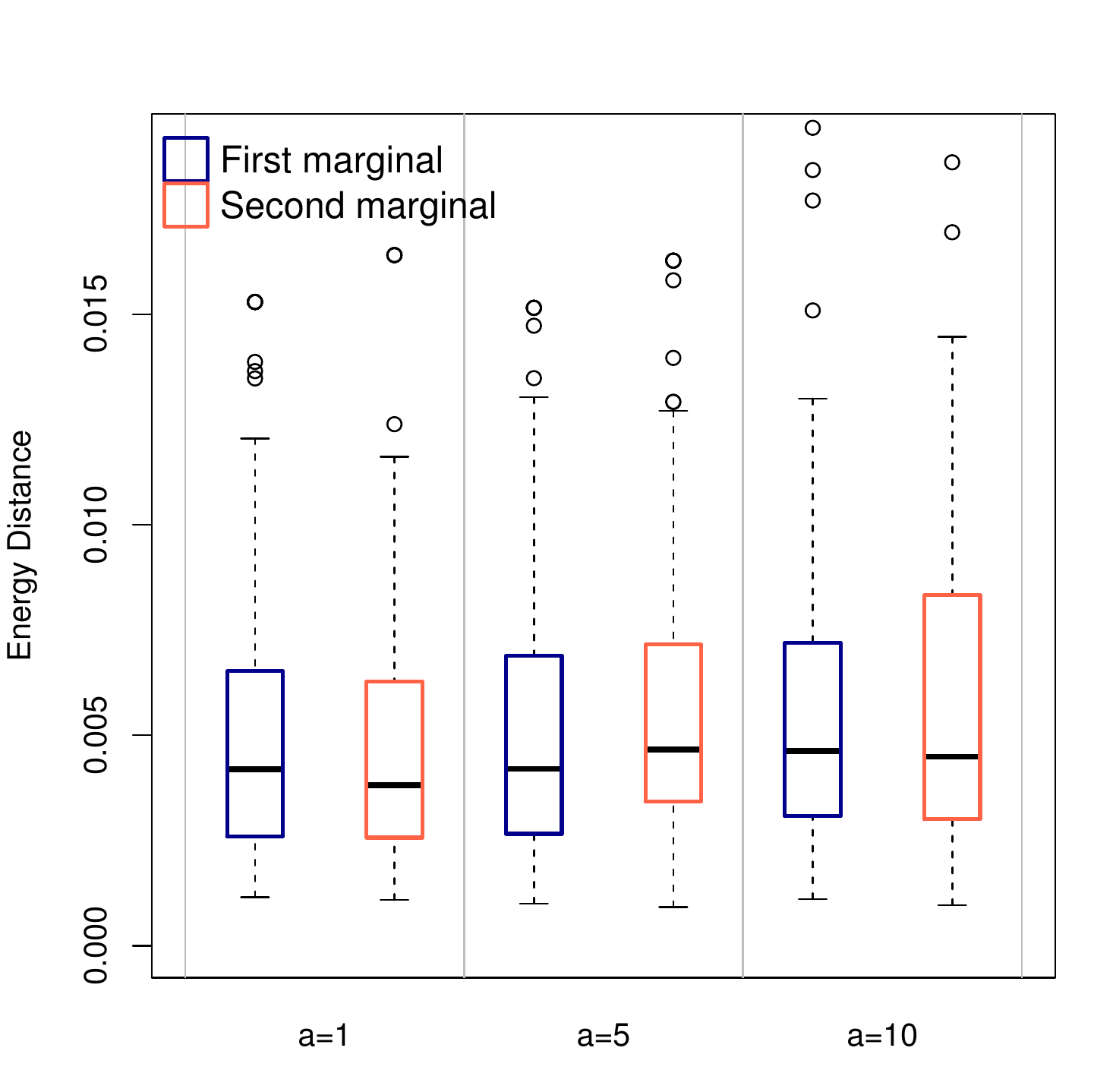} }}
    \caption{a: Boxplots of the Energy distance between $F_{true}=N_{2}(\mathbf{0}_{2},A_{2})$ and $F^{\ast}$ for $a=1,\,5,\,10$ and $n=1000$. b: Boxplots of the marginal distributions.}%
    \label{box.a.norm}%
\end{figure}
\begin{figure}[ht]
    \centering
    \subfloat[$F_{true}$: True Distribution.]{{\includegraphics[width=5.2cm]{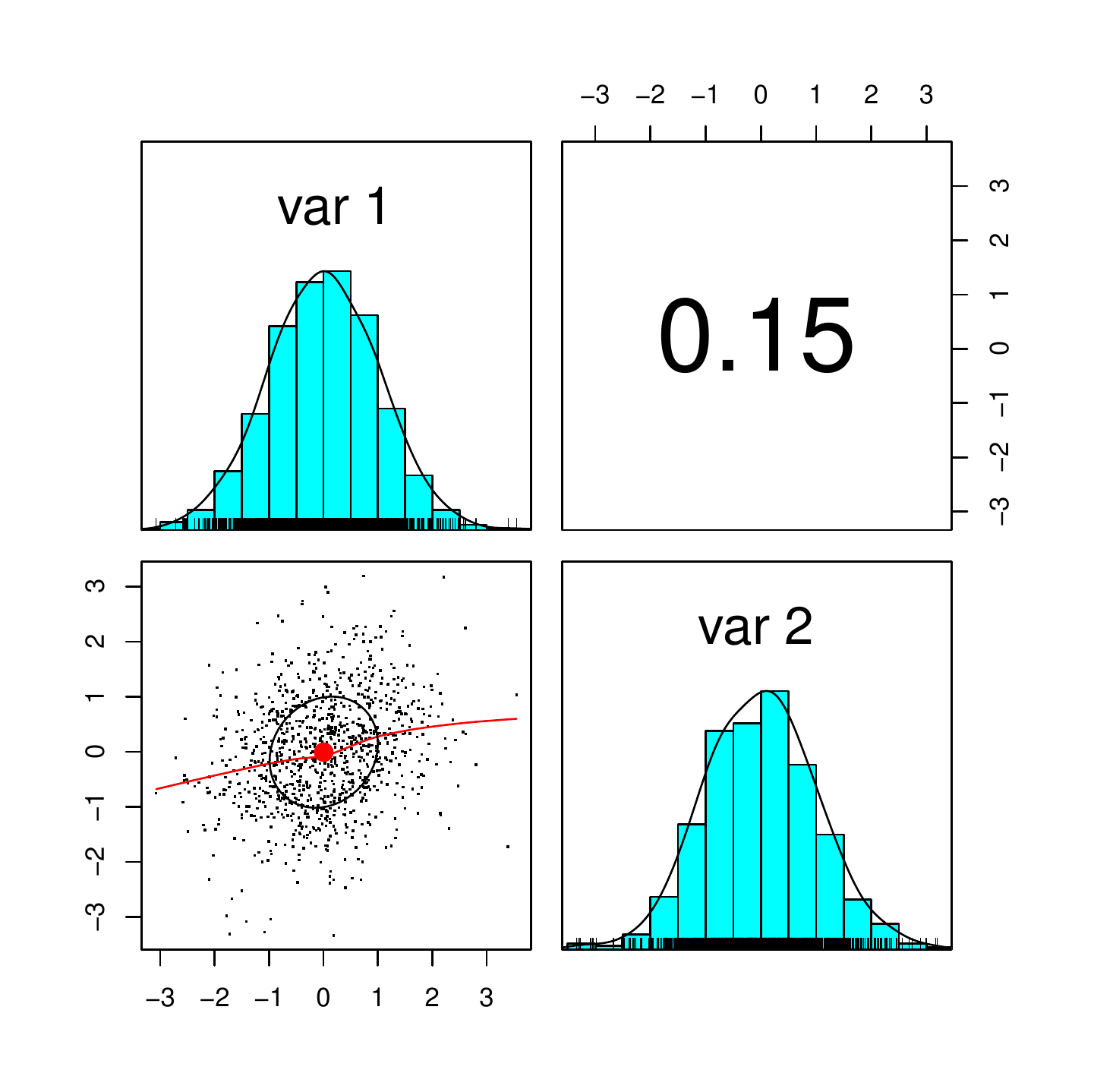} }}%
    \qquad
    \subfloat[$F^{\ast}$: Posterior-based model.]{{\includegraphics[width=5.2cm]{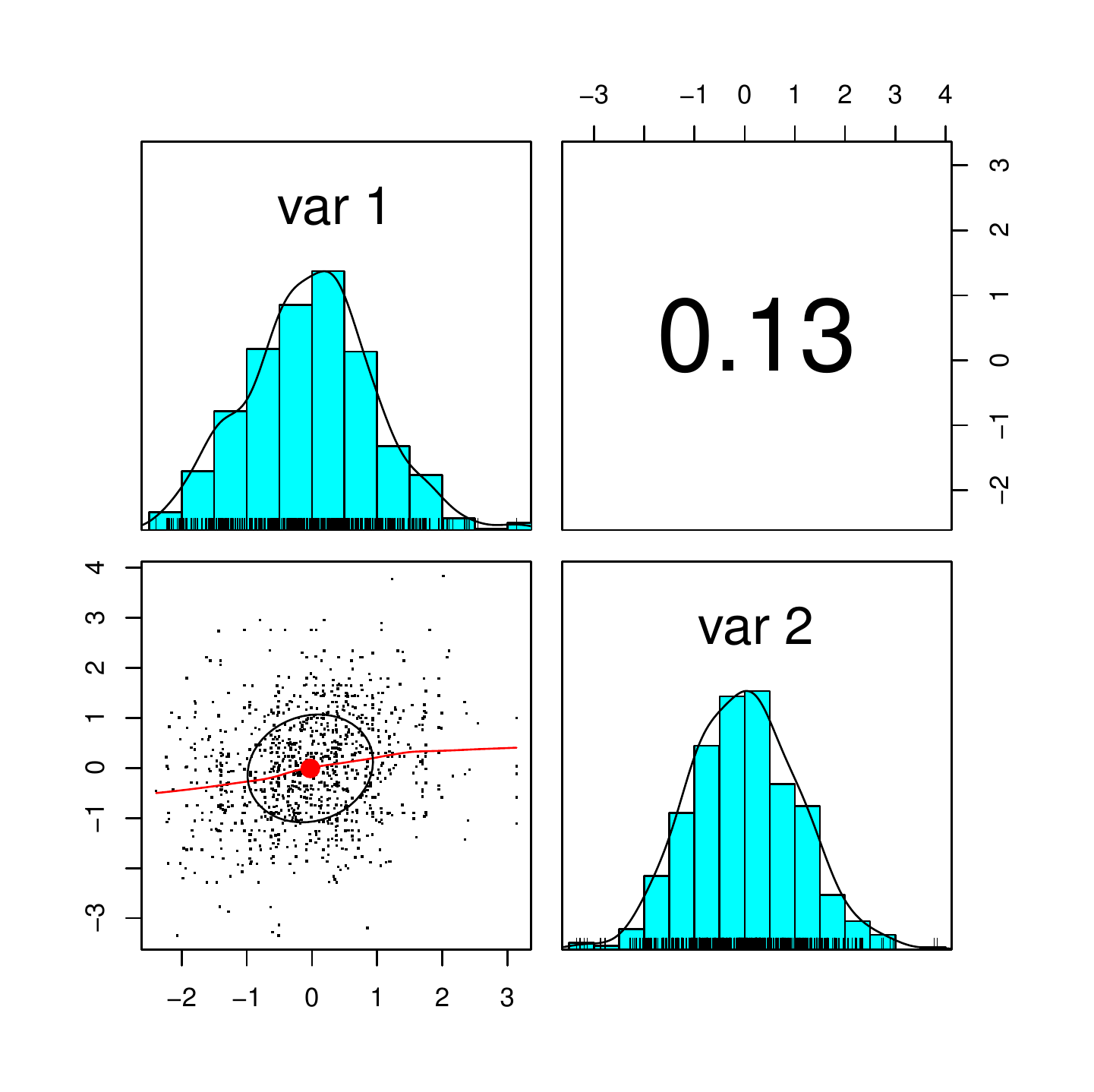} }}
    \caption{Marginal densities of $F_{true}=N_{2}(\mathbf{0}_{2},A_{2})$ and its posterior-based model with $n=1000$ based on the Kendall's $\tau$ for $a=1$.}%
    \label{fig 1}%
\end{figure}

Next, we inspect the effect of choosing different correlation coefficients such as the Gaussian rank, the Kendall's $\tau$ and the Spearman's $\rho$ on the posterior-based model. Consider $(\mathcal{P}_{VII}(1,1,1))^{2}$ and $N_{2}(\mathbf{0}_{2},A_{2})$ as two true distributions (consult  Table \ref{notation} in Appendix A for the notations). Table \ref{correlation} reports the results for $a=1$. Note that, the $\mathsf{R}$ package \textbf{rococo} is used to estimate $R^{\ast}$ based on the Gaussian rank correlation coefficients. It follows from Table \ref{correlation} that the performance of the methodology is approximately the same for different correlation coefficients.

\setlength{\extrarowheight}{.1mm}
\begin{table}[ht]
\centering
\setlength{\tabcolsep}{3 mm}
\caption{$\overline{d}_{\mathcal{E},N}(F^{\ast},F_{true})$ based on the Gaussian rank, Kendall, and Spearman correlation coefficients for $n=1000$ and $a=1$.}\label{correlation}
\scalebox{0.86}{
\begin{tabular}{c|ccc}
\toprule
\bfseries True distribution &\bfseries Gaussian rank &\bfseries Kendall's $\tau$ &\bfseries Spearman's $\rho$  \\
\hline
\scalebox{1.1}{$(\mathcal{P}_{VII}(1,1,1))^{2}$} &\scalebox{1.1}{0.00557} &\scalebox{1.1}{0.00528} &\scalebox{1.1}{0.00567}  \\
\scalebox{1.1}{$N_{2}(\mathbf{0}_{2},A_{2})$} &\scalebox{1.1}{0.00542} &\scalebox{1.1}{0.00524} &\scalebox{1.1}{0.00556} \\
\bottomrule
\end{tabular}
}
\end{table}
\subsection{Checking MVN Based on the BSPGC Approach}
The proposed normality test is illustrated through some interesting examples discussed in Henze and Visagie (2019). Note that, $NMIX1$ is a skewed heavy-tailed and $NMIX2$ is a symetric heavy-tailed distribution. Also, $(\mathcal{P}_{VII}(1,1,r))^{2}$, for $r\geq 10$ is a symetric distribution and has very similar behavior with a bivariate normal distribution. For a given sample of size  $n=50$, generated from distributions in Table \ref{T1}, the bivariate normality assumption is checked. For all cases, we set $N=r=1000$ and $M=20$ in Algorithm 3. To study the sensitivity of the approach, various values of $a$ are considered. The results of the proposed test are reported in Table \ref{T1}. The results are also compared to the Energy (E)-test (Sz\'{e}kely and Rizzo, 2013). Reminding that we want $RB>1$ and the strength close to 1 when $\mathcal{H}%
_{0}$ is true and $RB<1$ and the strength close to 0 when $\mathcal{H}_{0}$ is
false,  it is seen from  Table \ref{T2} that the proposed test has an excellent performance to accept or reject the bivariate normality assumption. The type I error and the power of the test are also reported in Table \ref{T2}. They show that the proposed test is powerful in both accepting and rejecting $\mathcal{H}_{0}$.

The next example uses a real data set.

\noindent\textbf{Real data example} (\emph{Swiss Heads}): In this example, we consider the data of six readings on the dimensions of the heads of 200 twenty year old Swiss soldiers given by Flury and Riedwyl (1988). The variables are minimal frontal breadth, breadth of angulus mandibulae, true facial height, length from glabella to apex nasi, length from tragion to nasion, and length from tragion to gnathion. The problem is to assess the six-variate normality assumption for this data set. The E-test's p-value is $2.2\times10^{-16}$, which shows strong evidence to reject the six-variate normality assumption. The proposed test presents $RB=0$ and strength$=0$ based on the Kendall's $\tau$ and $a=1$ which follows the methodology also presents strong evidence to reject the six-variate normality assumption.

We end this subsection by investigating   the effect of the prior-data conflict on the approach. This is in fact highlights the effect of the choice of $H$ in the prior-based model. For this, consider the results of MVN test when $F_{true}=(Exp(0.5))^{2}$ for different choices of $H$ in Table \ref{prior-data}. Clearly, when $H=F_{\boldsymbol{\theta}_{\mathbf{x}}}$ the results are correct; otherwise,  they  are incorrect.  Another concern is to check the effect of the double use of the data by considering $H$ as $F_{\boldsymbol{\theta}_{\mathbf{x}}}$ in the prior distance. Particularly, Table \ref{priori-invariant} gives the mean of the prior distance $\overline{d}_{\mathcal{E},N,a}(F_{N},H)$ for various choices of $H$. It is obvious from this table that the prior distance is invariant with respect to the choice of $H$.
\setlength{\extrarowheight}{0.1 mm}
\begin{table}[h!]
\centering
\setlength{\tabcolsep}{.9 mm}
\caption{Relative belief ratios and strength (str)-s for testing the bivariate normality assumption
with various alternatives and choices of $a$ based on the Kendall's $\tau$ with $n=50$.}\label{T1}
\scalebox{0.86}{
\begin{tabular}{ccccc||ccccc}
\toprule
{\parbox{22mm}{\center{\bfseries{Alternative\hspace{4mm} distribution}}}} &\bfseries $a$& $RB$&\bfseries  $ str$&{\parbox{22mm}{\center{\bfseries{E-test's\hspace{4mm} p-value}}}}&{\parbox{22mm}{\center{\bfseries{Alternative\hspace{4mm} distribution}}}} &\bfseries $a$& $RB$&\bfseries  $ str$&{\parbox{22mm}{\center{\bfseries{E-test's\hspace{4mm} p-value}}}}\\[.43 cm] \hline
 &\scalebox{1.1}{1} & \scalebox{1.1}{3.54} & \scalebox{1.1}{0.823} & &&\scalebox{1.1}{1} & \scalebox{1.1}{0.04} & \scalebox{1.1}{0.005} & \\[-.07cm]
\scalebox{1.1}{$N_{2}(\mathbf{0}_{2},I_{2})$}&\scalebox{1.1}{5}& \scalebox{1.1}{3.26} & \scalebox{1.1}{0.832} &\scalebox{1.1}{0.8794}&\scalebox{1.1}{$E(0.5)\otimes E(0.25)$}&\scalebox{1.1}{5} & \scalebox{1.1}{0.00} & \scalebox{1.1}{0.000} &\scalebox{1.1}{$2.2\times10^{-16}$}\\[-.07cm]
&\scalebox{1.1}{10}  & \scalebox{1.1}{2.48} & \scalebox{1.1}{0.997}  & &&\scalebox{1.1}{10}  & \scalebox{1.1}{0.00} & \scalebox{1.1}{0.000} & \\[-.01cm]\hline

 & \scalebox{1.1}{1}& \scalebox{1.1}{3.76} & \scalebox{1.1}{0.999}& &  &\scalebox{1.1}{1} & \scalebox{1.1}{0.80} & \scalebox{1.1}{0.110} & \\[-.07cm]
\scalebox{1.1}{$N_{2}(\mathbf{0}_{2},A_{2})$}&\scalebox{1.1}{5} &  \scalebox{1.1}{2.82}  & \scalebox{1.1}{0.859} & \scalebox{1.1}{0.8442}&\scalebox{1.1}{$\mathcal{S}^{2}(LN(0,0.25))$}&\scalebox{1.1}{5} & \scalebox{1.1}{0.62} & \scalebox{1.1}{0.151} &\scalebox{1.1}{$2.2\times10^{-16}$}\\[-.07cm]
&\scalebox{1.1}{10} & \scalebox{1.1}{2.22}      & \scalebox{1.1}{0.884} & &&\scalebox{1.1}{10}  & \scalebox{1.1}{0.60} & \scalebox{1.1}{0.240}  & \\[-.01cm]\hline

&\scalebox{1.1}{1} & \scalebox{1.1}{0.18} & \scalebox{1.1}{0.017} && &\scalebox{1.1}{1}  & \scalebox{1.1}{2.92} & \scalebox{1.1}{0.854} & \\[-.07cm]
\scalebox{1.1}{$LN_{2}(\mathbf{0}_{2},B_{2})$}&\scalebox{1.1}{5} & \scalebox{1.1}{0.18} & \scalebox{1.1}{0.033} &\scalebox{1.1}{$2.2\times10^{-16}$}&\scalebox{1.1}{$(\mathcal{P}_{VII}(1,1,10))^{2}$}&\scalebox{1.1}{5} & \scalebox{1.1}{2.74} & \scalebox{1.1}{0.999} &\scalebox{1.1}{0.7035}\\[-.07cm]
&\scalebox{1.1}{10}  & \scalebox{1.1}{0.06} & \scalebox{1.1}{0.007}  && &\scalebox{1.1}{10} & \scalebox{1.1}{2.26} & \scalebox{1.1}{0.878}  & \\[-.01cm]\hline

&\scalebox{1.1}{1} & \scalebox{1.1}{0.10} & \scalebox{1.1}{0.000} & &&\scalebox{1.1}{1} & \scalebox{1.1}{0.34} & \scalebox{1.1}{0.017} & \\[-.07cm]
\scalebox{1.1}{$NMIX1$}&\scalebox{1.1}{5} & \scalebox{1.1}{0.02} & \scalebox{1.1}{0.002} &\scalebox{1.1}{0.0050}&\scalebox{1.1}{$(\chi^{2}_{5})^{2}$}&\scalebox{1.1}{5} & \scalebox{1.1}{0.22} & \scalebox{1.1}{0.026} &\scalebox{1.1}{0.03518}\\[-.07cm]
&\scalebox{1.1}{10}  & \scalebox{1.1}{0.00} & \scalebox{1.1}{0.000}  & &&\scalebox{1.1}{10} & \scalebox{1.1}{0.04} & \scalebox{1.1}{0.002}  & \\[-.01cm]\hline

&\scalebox{1.1}{1} & \scalebox{1.1}{3.64} & \scalebox{1.1}{1.000} && &\scalebox{1.1}{1} & \scalebox{1.1}{0.58} & \scalebox{1.1}{0.099} & \\[-.07cm]
\scalebox{1.1}{$NMIX2$}&\scalebox{1.1}{5} & \scalebox{1.1}{3.32} & \scalebox{1.1}{0.834} &\scalebox{1.1}{0.8744}&\scalebox{1.1}{$N(0,1)\otimes \chi^{2}_{5}$}&\scalebox{1.1}{5} & \scalebox{1.1}{0.40} & \scalebox{1.1}{0.070} &\scalebox{1.1}{0.4020}\\[-.07cm]
&\scalebox{1.1}{10}  & \scalebox{1.1}{2.34} & \scalebox{1.1}{0.874} &&&\scalebox{1.1}{10} & \scalebox{1.1}{0.30} & \scalebox{1.1}{0.095}  & \\[-.01cm]\hline

&\scalebox{1.1}{1} & \scalebox{1.1}{0.80} & \scalebox{1.1}{0.210} && &\scalebox{1.1}{1} & \scalebox{1.1}{0.63} & \scalebox{1.1}{0.160} & \\[-.07cm]
\scalebox{1.1}{$t_{5}(\mathbf{0}_{2},I_{2})$}&\scalebox{1.1}{5} & \scalebox{1.1}{0.60} & \scalebox{1.1}{0.181} &\scalebox{1.1}{0.0452}&\scalebox{1.1}{$N(0,1)\otimes t_{3}$}&\scalebox{1.1}{5} & \scalebox{1.1}{0.44} & \scalebox{1.1}{0.092} &\scalebox{1.1}{0.0502}\\[-.07cm]
&\scalebox{1.1}{10} & \scalebox{1.1}{0.20} & \scalebox{1.1}{0.010}  & &&\scalebox{1.1}{10} & \scalebox{1.1}{0.40} & \scalebox{1.1}{0.020}  & \\[-.01cm]\bottomrule

\end{tabular}
}
\end{table}

\setlength{\extrarowheight}{.1mm}
\begin{table}[ht]
\centering
\setlength{\tabcolsep}{7.5 mm}
\caption{The proportion of rejecting (POR) $\mathcal{H}_{0}$ out of 1000 replications based on the Kendall's $\tau$ for $a=1$ and sample of size $n=50$.}\label{T2}
\scalebox{0.86}{
\begin{tabular}{cc||cc}
\toprule\\[-.4cm]
\bfseries Distribution & \bfseries POR $\mathcal{H}_{0}$&\bfseries Distribution&\bfseries POR $\mathcal{H}_{0}$\\\hline
\scalebox{1.2}{$N_{2}(\mathbf{0}_{2},I_{2})$}&\scalebox{1.2}{0.057}$^{\dagger}$& \scalebox{1.2}{$NMIX1$} &\scalebox{1.2}{0.794}$^{\ddagger}$ \\

\scalebox{1.2}{$N_{2}(\mathbf{0}_{2},A_{2})$}&\scalebox{1.2}{0.070$^{\dagger}$}&  \scalebox{1.2}{$NMIX2$}  & \scalebox{1.2}{0.077$^{\ddagger}$} \\

\scalebox{1.2}{$LN_{2}(\mathbf{0}_{2},B_{2})$}&\scalebox{1.2}{0.801$^{\ddagger}$}&  \scalebox{1.2}{$(\mathcal{P}_{VII}(1,1,10))^{2}$}  & \scalebox{1.2}{0.106$^{\ddagger}$} \\

\scalebox{1.2}{$\mathcal{S}^{2}(LN(0,0.25))$}&\scalebox{1.2}{0.798$^{\ddagger}$}& \scalebox{1.2}{$(\chi^{2}_{5})^{2}$} & \scalebox{1.2}{0.782$^{\ddagger}$} \\

\scalebox{1.2}{$t_{5}(\mathbf{0}_{2},I_{2})$}&\scalebox{1.2}{0.499$^{\ddagger}$}&  \scalebox{1.2}{$N(0,1)\otimes \chi^{2}_{5}$}  & \scalebox{1.2}{0.657$^{\ddagger}$} \\

\scalebox{1.2}{$E(0.5)\otimes E(0.25)$}&\scalebox{1.2}{0.999$^{\ddagger}$}&  \scalebox{1.2}{$N(0,1)\otimes t_{3}$}  & \scalebox{1.2}{0.551$^{\ddagger}$} \\\bottomrule

\end{tabular}
}
\centering
\begin{tablenotes}
      \small
      \item \fontsize{8.6}{8}\selectfont{$^{\dagger}$ Type I error of the proposed test.}
      \item $^{\ddagger}$ Power of the proposed test.
    \end{tablenotes}
\end{table}
\setlength{\extrarowheight}{0 mm}
\begin{table}[H]
\centering
\setlength{\tabcolsep}{5 mm}
\caption{RB(Strength) of a sample of size 50 generated from $(E(0.5))^{2}$ when there is prior-data conflict (a tiny overlap between the effective support regions).}\label{prior-data}
\scalebox{0.86}{
\begin{tabular}{c|ccc}
\toprule
$H$ &$a=1$ &$a=5$ &$a=10$  \\
\hline
\scalebox{1.1}{$F_{\boldsymbol{\theta}_{\mathbf{x}}}$} &\scalebox{1.1}{0.080(0.004)} &\scalebox{1.1}{0.040(0.006)} &\scalebox{1.1}{0.00(0.000)}  \\
\scalebox{1.1}{$N_{2}(\mathbf{x},I_{2})$} &\scalebox{1.1}{18.46(1.000)} &\scalebox{1.1}{9.180(0.989)} &\scalebox{1.1}{9.08(0.969)} \\
\scalebox{1.1}{$N_{2}(\mathbf{3}_{2},S_{\mathbf{x}})$} &\scalebox{1.1}{19.00(1.000)} &\scalebox{1.1}{10.71(1.000)} &\scalebox{1.1}{8.32(0.582)} \\
\bottomrule
\end{tabular}
}
\end{table}
\setlength{\extrarowheight}{.5mm}
\begin{table}[ht]
\centering
\begin{small}
\setlength{\tabcolsep}{5 mm}
\caption{The mean of the prior distance $\overline{d}_{\mathcal{E},N,a}(F_{N},H)$ for various choices of $H$ when $\mathbf{x}_{2\times 50}\sim (E(0.5))^{2}$ based on the Kendall's $\tau$ with $N=1000$ and $a=1$.}\label{priori-invariant}
\scalebox{.86}{
\begin{tabular}{cc||cc}
\toprule
\scalebox{1.2}{$H$} &\scalebox{1.2}{$\overline{d}_{\mathcal{E},N,a}(F,H)$} & \scalebox{1.2}{$H$} & \scalebox{1.2}{$\overline{d}_{\mathcal{E},N,a}(F,H)$} \\
\hline
\scalebox{1.2}{$F_{\boldsymbol{\theta}_{\mathbf{x}}}$} & \scalebox{1.2}{0.00672} & \scalebox{1.2}{$N_{2}(\mathbf{3}_{2},A_{2})$} & \scalebox{1.2}{0.00633}  \\
\scalebox{1.2}{$N_{2}(\mathbf{0}_{2},I_{2})$} & \scalebox{1.2}{0.00661} & \scalebox{1.2}{$NMIX2$} & \scalebox{1.2}{0.00658} \\
\bottomrule
\end{tabular}
}
\end{small}
\end{table}

\section{Concluding Remarks}
A BSPGC approach and its application to the MVN test have been suggested. In this procedure, a Gaussian copula model has been utilized to induce the dependence structure of the underlying multivariate distribution $F_{true}$. The Dirichlet process then has been constructed on the unknown margins of $F_{true}$ to define the prior-based and posterior-based models, respectively. The test has been developed by using the relative belief ratio for comparing the concentration of the distribution of the distance between the posterior-based model and the null distribution versus the concentration of the distribution of the distance between the prior-based model and the null distribution at zero. The Energy distance has been applied to compute distances as an appropriate tool especially in high dimensional problems. The methodology has been examined by a simulation study to clarify its excellent performance. Finally, application of the test including a real data example has been presented. A main advantage of the procedure is that it takes into account the dependence structure of the data in the MVN test. The extension of the procedure to different areas of the multivariate data analysis by considering various families of copula will be a part of a future research work.

\begin{appendices}
\section{Relevant Notations}
\begin{small}
\setlength{\extrarowheight}{.7mm}
\begin{table}[ht]
\centering
\caption{Description of notations}\label{notation}
\scalebox{.81}{
\begin{tabular}{l}
\toprule
\textbf{Notation:} Description  \\
\hline
1. $\mathbf{c}_{2}:=(c,c)^{T}$, $I_{2}:=\bigl(\begin{smallmatrix} 1&0\\0&1\end{smallmatrix}\bigr)$, $A_{2}:=\bigl(\begin{smallmatrix} 1&0.2\\0.2&1\end{smallmatrix}\bigr)$ and $B_{2}:=\bigl(\begin{smallmatrix} 0.25&0.2\\0.2&0.025\end{smallmatrix}\bigr)$. \\
{\parbox{16cm}{2. $E(\lambda)$: An exponentional distribution with rate $\lambda$.}}\\
{\parbox{16cm}{3. $t_{r}$: A $t$-Studen distribution with $r$ degrees of freedom.}}\\
{\parbox{16cm}{4. $B(\alpha,\beta)$: A Beta distribution with shape 1 parameter $\alpha$ and shape 2 parameter $\beta$.}}\\
{\parbox{16cm}{5. $\chi_{r}$: A chi-square distribution with $r$ degrees of freedom.}}\\
{\parbox{16cm}{6. $\mathcal{P}_{VII}(1,1,r)^{\star}$: A pearson type $VII$ (aka $t$-Student) distribution with location parameter 1, scale parameter 1 and  $r$ degrees of freedom.}}\\
{\parbox{16cm}{7. $F_{1}\otimes F_{2}$: A bivariate distribution with two independent marginal distributions $F_{1}$ and $F_{2}$.}}\\
{\parbox{16cm}{8. $t_{r}(\mathbf{0}_{2},I_{2})^{\dagger}$: A bivariate $t$-student distribution with location parameter $\mathbf{0}_{2}$, scale parameter $I_{2}$ and $r$ degrees of freedom.}}\\
{\parbox{16cm}{9. $LN_{2}(\mathbf{0}_{2},B_{2})^{\ddagger}$: A bivariate lognormal distribution with mean vector $\mathbf{0}_{2}$ and covariance matrix $B_{2}$.}}\\
{\parbox{16cm}{10. $\mathcal{S}^{2}(LN(0,0.25))^{\dagger}$: A bivariate spherical distribution with lognormal distribution $LN(0,0.25)$ for radii.}}\\
11. $NMIX1^{\dagger}$: $0.9N_{2}(\mathbf{0}_{2},I_{2})+0.1N_{2}(\mathbf{3}_{2},I_{2})$\\
12. $NMIX2^{\dagger}$: $0.9N_{2}(\mathbf{0}_{2},A_{2})+0.1N_{2}(\mathbf{0}_{2},I_{2})$.
\\
\bottomrule
\end{tabular}
}
\centering
\begin{tablenotes}
      \small
      \item \fontsize{8}{8}\selectfont{$^{\star}$ Required $\mathsf{R}$ package: PearsonDS.}
      \item $^{\dagger}$ Required $\mathsf{R}$ package: distrEllipse.
       \item $^{\ddagger}$ Required $\mathsf{R}$ package: compositions.
    \end{tablenotes}
    \end{table}
    \end{small}

\end{appendices}

\end{document}